\documentclass[12pt, reqno]{amsart}
\usepackage{color}
\usepackage{graphicx}
\usepackage{psfig}
\usepackage{amscd}
\usepackage{amsmath}
\usepackage{amsfonts}
\usepackage{amssymb}
\usepackage{latexsym}
\usepackage[mathscr]{eucal}

\theoremstyle{plain}

\newtheorem{lemma}{Lemma}
\newtheorem{theorem}{Theorem}
\newtheorem{corollary}{Corollary}[theorem]
\theoremstyle{definition}
\renewcommand{\refname}{\textbf{References}}
\newtheorem{remark}{Remark}
\def\d#1{{#1\kern-0.4em\char"16\kern-0.1em}}
\def\D#1{{\raise0.2ex\hbox{-}\kern-0.4em #1}}
\begin{document}
\allowdisplaybreaks
\title[UNIVERSAL SAMPLING TRUNCATION ERROR UPPER BOUNDS]{UNIVERSAL TRUNCATION ERROR UPPER BOUNDS IN SAMPLING RESTORATION\footnote{ This is an Author's Accepted Manuscript of an article published in the Georgian Mathematical Journal. Vol.~17, No. 4. (2010), 765--786. The final publication is available at De Gruyter.  DOI: 10.1515/gmj.2010.033}}
\author{Andriy Ya. Olenko}
\address{ Department of Mathematics and Statistics, La Trobe University, Victoria 3086, Australia}
\email{a.olenko@latrobe.edu.au}

\author{Tibor K. Pog\'any}
\address{Faculty of Maritime Studies, University of Rijeka,
51000 Rijeka, Studentska 2, Croatia}
\email{\tt poganj@brod.pfri.hr}


\begin{abstract} Universal (pointwise uniform and time shifted) truncation error upper bounds
are presented for the Whittaker--Kotel'nikov--Shannon (WKS) sampling restoration sum for Bernstein
function classes $B_{\pi,d}^q,\, q>1,\, d\in \mathbb N$, when the decay rate of the sampled functions is unknown.
The case of regular sampling is discussed. Extremal properties of related series of sinc functions are investigated.
\end{abstract}

\maketitle

\section{\bf Introduction}
Let $\mathsf X$ be a normed space and assume that the structure of $\mathsf X$ admits the sampling restoration
procedure
  \begin{equation} \label{1}
      f(\mathbf{x}) = \sum_{\mathbf{n}\in \mathbb Z^d} f(t_\mathbf{n})S(\mathbf{x},t_\mathbf{n}) \qquad
      \big( f\in \mathsf X\big)
  \end{equation}
where $\mathbf{x} \in \mathbb R^d$, $\mathfrak T :=
\{t_\mathbf{n}\}_{\mathbf{n} \in \mathbb Z^d} \subset \mathbb R^d$ is the set of sampling points and $S(\cdot,\cdot)$
is the sampling function. This formula is one of the basic tools in signal processing.

In direct numerical implementations we consider the truncated variant of (\ref{1}), that is
   \[ Y_{\mathfrak J}(f;\mathbf{x}) = \sum_{\mathbf{n}\in \mathfrak J} f(t_\mathbf{n})S(\mathbf{x},t_\mathbf{n})
                                      \qquad \big( \mathfrak J \subset \mathbb Z^d\big)\, , \]
where the index set $\mathfrak J$ is necessarily finite in applications. Namely, restoring the continuous signal from
discrete samples or assessing the information lost in the sampling process are the fundamental problems in sampling
and interpolation theory.

The usual procedure is to estimate the truncation error
   \begin{equation} \label{3}
      \| T_{\mathfrak J}(f;\mathbf{x})\| :=  \|f(\mathbf{x}) -
          Y_{\mathfrak J}(f;\mathbf{x})\| \le \varphi_{\mathfrak J}(f;\mathbf{x}),
   \end{equation}
where $\| \cdot\|$ could denote any suitable norm and $\varphi_{\mathfrak J}$ denotes the truncation error upper bound.
Simple truncation error upper bounds are the main tools in numerical implementations, when they do not contain infinite
products and/or unknown function values. However, a sharp truncation error upper bound enables pointwise, almost sure,
uniform {\em etc.} convergence of the approximating sequence $Y_{\mathfrak J}(f;\mathbf{x})$ to the initial $f \in \mathsf X$ when
$|\mathfrak J| \to \infty$.

In their recent article \cite{OP2} the authors established sharp upper bounds for interpolation remainders of the multidimensional
Paley--Wiener class functions in the $L^2$-space, in the case of finite, regular (equidistantly sampled) WKS sampling sums and consequent extremal functions
are given \cite[Theorems 1, 2]{OP2}. Also, truncation error analysis and convergence rate is provided in weak Cram\'er
class random fields \cite[Theorems 3, 4]{OP2}. Here we generalize in some fashion certain our findings to Bernstein function class in the $L^p$-spaces.

The main aim of this paper is to discuss $T_{\mathfrak J}(f; \mathbf x)$ when $f\in L^p$ is coordinatewise exponentially
bounded entire function, that is, belongs to suitable Bernstein functions class, in different situations by obtaining
pointwise upper bounds valid for all $\mathbf x$ belonging to the signal domain \cite{OP2}. We call this kind of upper
bound {\em universal}.

All numerical computations in the paper were done using {\it Mathematica}~5.0\,.

\section{\bf Multidimensional Plancherel--P\'olya inequality}
The following few results will prove useful throughout. The first one is the multidimensional
analogue of the famous {\em Plancherel--P\'olya inequality} extended from the one--dimensional
case published in \cite{PP, PPII}.

Denote $\| \cdot \|_p$ the $L_p$--norm  (while $\| \cdot\|_\infty
\equiv \mathrm{ess}\,\, \sup |\cdot |)$ and let $L^p(\mathbb R)$ be the
class of all complex--valued functions whose restrictions to $\mathbb R$ have finite
$L_p$--norm. Assume $f \in L^r(\mathbb R),\, r>0$ to be of exponential type $\sigma>0$ and let
$\{ t_n\}_{n\in \mathbb Z}$ be a separated real sequence, i.e. such that $\inf_{n\neq m}|t_n-t_m|\ge \delta >0$.
Then we have \cite[Eq.~(76)]{PPII}
   \begin{equation} \label{PP}
      \sum_{n \in \mathbb Z} |f(t_n)|^r \le B \|f\|_r^r \, ,
   \end{equation}
where
   \[ \displaystyle
      B = \frac{8(e^{r\,\delta \sigma /2} -1)}{ r\, \pi \sigma\delta^2 } \, .\]
Formula \eqref{PP} is the celebrated Plancherel--P\'olya inequality. It could be mentioned that Boas
\cite{Boas0} established another estimate in the one--dimensional case under different assumptions on $\mathfrak T$
and recently Lindner published an estimate in the one--dimensional case when $p = 2$, \cite{lin}.

Now, we give the multidimensional analogue of the Plancherel--P\'olya inequality. Here, and in what follows
$B_{\boldsymbol \sigma,d}^r,\, r>0$ denotes the Bernstein class \cite{nik} of $d$--variable entire functions
of exponential type at most $\boldsymbol \sigma = (\sigma_1, \cdots, \sigma_d)$ coordinatewise whose restriction
to $\mathbb R^d$ is in $L^r(\mathbb R^d)$.

We will use the following Nikolski\u{\i}--type inequality of different dimensions. 
\begin{lemma}\label{nik_i}\cite[\S 3.4.2]{nik}
If $1\le r \le \infty$ and $1\le m \le d,$ then for each entire function $g(\mathbf z)\in L^r(\mathbb R^d)$
of exponential type $\boldsymbol \sigma$ the following inequality  holds
   \begin{align*}
      \Bigg(\int_{\mathbb R^m} & |g(z_1,\cdots, z_m, z_{m+1},\cdots,z_d)|^r\, {\rm d}z_1\cdots {\rm d}z_m\Bigg)^{1/r} \\
            &\le 2^{d-m}\Bigg(\prod_{k=m+1}^d \sigma_k\Bigg)^{1/r}\Bigg(\int_{\mathbb R^d}
                 |g(z_1,\cdots,z_d)|^r\, {\rm d}z_1\cdots {\rm d}z_d\Bigg)^{1/r}\,.
   \end{align*}
For fixed $d$ and $m$ and arbitrary $ \boldsymbol \sigma = (\sigma_1, \cdots, \sigma_d)$ this inequality is exact in the sense of order.
\end{lemma}

\begin{theorem}\label{th1} Let $f \in B_{\boldsymbol \sigma,d}^r,\, r\ge 1$ and let $\mathfrak T = \{ t_{\mathbf n}
=(t_{n_1},...t_{n_d})\}_{\mathbf n \in \mathbb Z^d},$ be real  separated sequence, that is
   \[ \inf_{n_\ell \neq m_\ell}|t_{n_\ell}-t_{m_\ell}|\ge \delta_\ell>0 \qquad \big(\ell = 1, \cdots,d\big).\]
Then
   \begin{equation} \label{PP-d}
      \sum_{\mathbf n \in \mathbb Z^d} |f(t_{\mathbf n})|^r \le \mathfrak B_{d,r} \|f\|_r^r \, ,
   \end{equation}
where
   \begin{equation}  \label{PP-dconst} \displaystyle
      \mathfrak B_{d,r} = \Big( \frac{8}{r \pi}\Big)^d \prod_{\ell=1}^d \frac{e^{r\,\delta_\ell \sigma_\ell /2} -1}
                    {\sigma_\ell \delta_\ell^2} \, .
   \end{equation}
\end{theorem}
\begin{proof} Take $d=2$; the proof will be identical in the case $d>2$. By assumption
$f \in B_{\boldsymbol \sigma,d}^r$ and bearing in mind Lemma~\ref{nik_i} which holds for $r\ge 1$, we conclude that
$f(\cdot,x_2)\in B_{\sigma_1,1}^r$and $f(x_1, \cdot)\in B_{\sigma_2,1}^r$. Therefore, we can apply \eqref{PP} coordinatewise
to $\sum_{\mathbf n \in \mathbb Z^d}|f(t_{\mathbf n})|^r$. Because $\{t_{n_\ell}\}_{n_\ell\in \mathbb Z}$ are separated
with $\delta_\ell,\, \ell=1,2$ we deduce
   \begin{align}
      \sum_{\mathbf n \in \mathbb Z^2} |f(t_{\mathbf n})|^r &=
              \sum_{n_1 \in \mathbb Z} \Bigg( \sum_{n_2 \in \mathbb Z}| f(t_{n_1},t_{n_2})|^r \Bigg) \nonumber\\
         &\le \frac{8(e^{r\,\delta_2 \sigma_2 /2} -1)}{ r\, \pi\sigma_2 \delta_2^2 }
              \sum_{n_1 \in \mathbb Z} \int_{\mathbb R}|f(t_{n_1},x_2)|^r\, {\rm d}x_2\, \nonumber \\  \label{delta2}
         &= \frac{8(e^{r\,\delta_2 \sigma_2 /2} -1)}{ r\, \pi\sigma_2 \delta_2^2 }
              \int_{\mathbb R} \Bigg( \sum_{n_1 \in \mathbb Z} |f(t_{n_1},x_2)|^r\Bigg) {\rm d}x_2 \, .
   \end{align}
The second subsequent application of (\ref{PP}) to (\ref{delta2}) yields
   \begin{equation} \label{delta1}
      \sum_{\mathbf n \in \mathbb Z^2} |f(t_{\mathbf n})|^r \le
            \frac{64(e^{r\,\delta_1 \sigma_1 /2} -1)(e^{r\,\delta_2 \sigma_2 /2} -1)}
                 {\sigma_1\sigma_2\, r^2\, \pi^2 (\delta_1\delta_2)^2 }
                 \int_{\mathbb R^2} |f(x_1,x_2)|^r {\rm d}x_1 {\rm d}x_2 = \mathfrak B_{2,r} \|f\|_r^r\, ,
\end{equation}
where in (\ref{delta1}) Fubini's theorem is used. So, the assertion  of the theorem is proved.
\end{proof}
\begin{remark} {\it In {\rm \cite[\S 46, p.148]{PPII}} the authors assume the condition
   \begin{equation} \label{Delta}
      \sum_{\ell =1}^d |t_{\mathbf n}^{(\ell)}-t_{\mathbf m}^{(\ell)}|^2 \ge \Delta^2
   \end{equation}
on $\mathfrak T$ and $p>0$, so that the multidimensional variant of \eqref{PP} holds. It is clear that our
separability assumption upon $\mathfrak T$ is stronger then \eqref{Delta}, where
   \[ \Delta^2 = \sum_{\ell=1}^d\delta_\ell^2 \]
could be taken. Our approach results in the explicit Plancherel--P\'olya--type constant \eqref{PP-dconst} such that plays
crucial role in establishing {\em explicit} (not necessarily {\em exact}) truncation error upper bounds in the multidimensional
Whittaker--Kotel'nikov--Shannon (WKS) sampling theorem, being a part of numerical routines in computation. }
\end{remark}

Multidimensional Plancherel--P\'olya inequalities can be found e.g. in Triebel's book \cite{T}; also,
during last years several additional very far going generalizations of the multidimensional
Plancherel--P\'olya inequality were obtained, among others in \cite{Han}, \cite{Pes}.  Unfortunately,
no explicit estimates of the Plancherel--P\'olya constant have appeared neither in these articles, nor in articles referenced therein
including the `ancestor' article \cite{PPII}.

\section{\bf On extrema of  $\sum |{\rm sinc}(x-n)|^p$ and their properties}
Another frequently needed mathematical tool concerns upper bound estimates for $\sum |{\rm sinc}(x-n)|^p,\ p>1$.
Here ${\rm sinc}(x) := \frac{\sin (\pi x)}{\pi x}$ for $x \neq 0$ and ${\rm sinc}(0) := 1$.
\begin{theorem}\label{the0} There holds
   \[ \sum_{n\in\mathbb Z}|{\rm sinc}(x-n)|^p\le \mathfrak C_p \qquad \big(x\in \mathbb R\big)\]
where
   \[ \mathfrak C_p:= \begin{cases} \displaystyle
                                1+\left(\frac{2}{\pi}\right)^p\frac{p}{p-1} & \ 1<p<2 \\
                                1 & \ p\ge 2
                      \end{cases}\quad . \]
\end{theorem}
The statement is the straightforward consequence of lemmata \cite[Lemma 11.2]{hig} and \cite[Lemma 2.4]{li}, so we omit the proof.

Let $\mathfrak J=\mathfrak J_x := \{n:\, |x - n|\le N\},$ $ x\in \mathbb R,$ $N\in\mathbb N$; accordingly let us
define
   \[ \mathfrak h_{p,N}(x) := \sum_{n \not\in \mathfrak J_x}|{\rm sinc}(x-n)|^p \, . \]
It was shown \cite{OP1} that for arbitrary $N\in \mathbb N$ the function $\max \mathfrak h_{2,N}(x) =
\mathfrak h_{2,N}(1/2)$. We begin by showing that this is not true for $p$ and $N$ in general
-- compare the Figures 1,2 below.
\begin{center}
\begin{minipage}{6cm}
{\psfig{figure=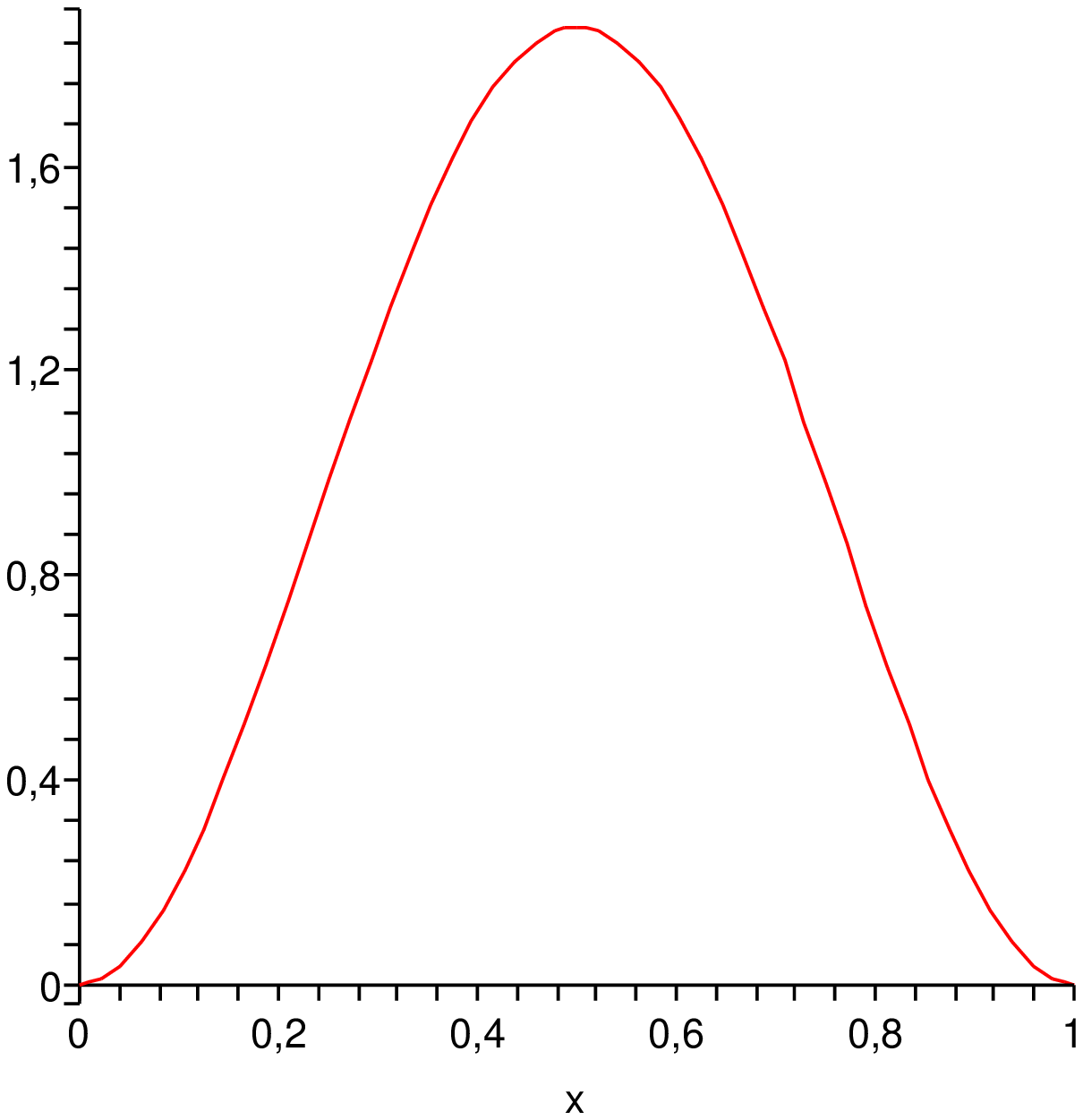,width=6.4cm}} \hspace{0.5cm}{\quad \qquad Figure 1.\ \  $\mathfrak h_{2,2}(x)$}
\end{minipage}
\begin{minipage}{6cm}
{\psfig{figure=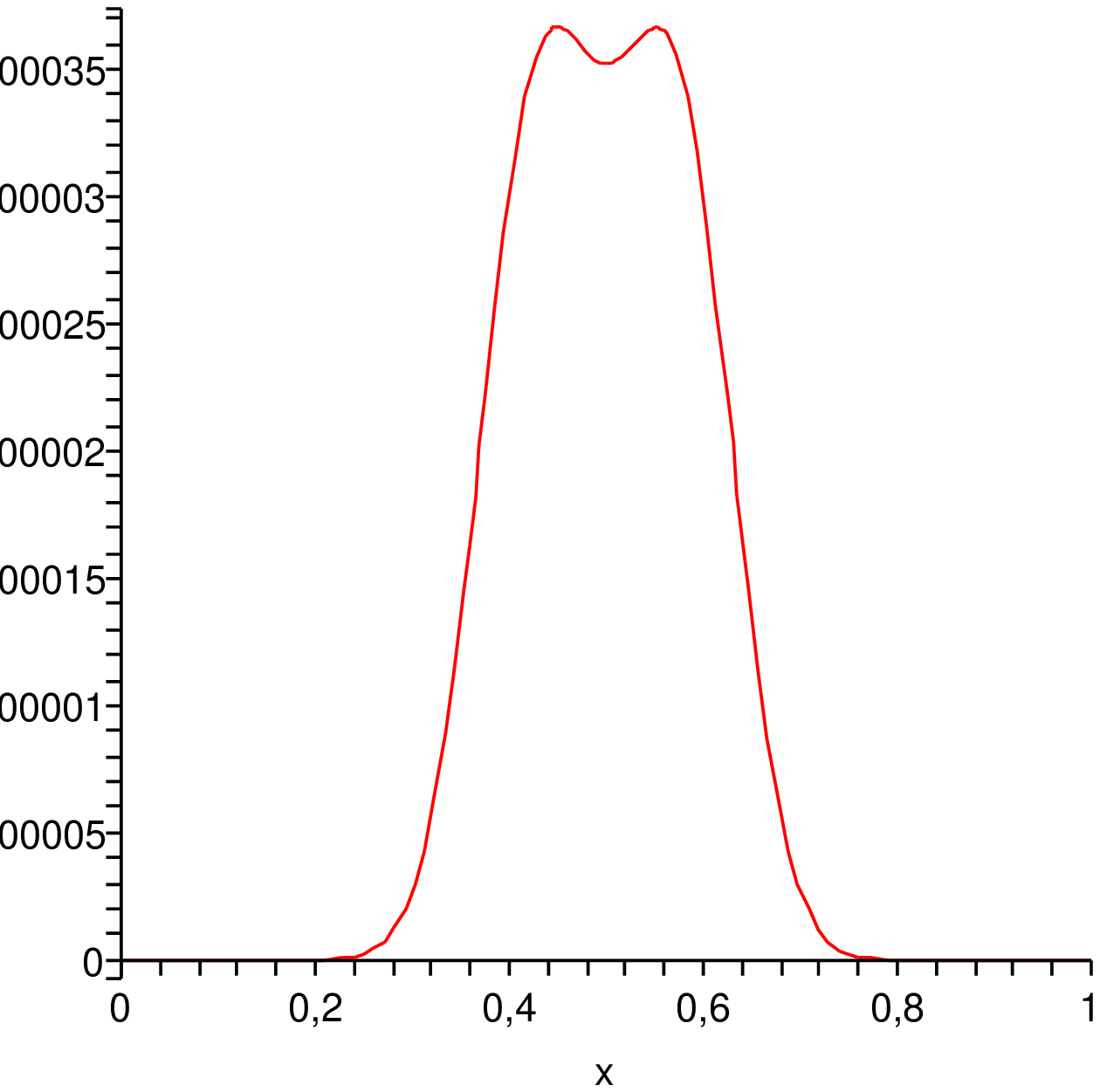,width=6.5cm}} \hspace{0.5cm}{\qquad \quad Figure 2.\ \ $\mathfrak h_{27,2}(x)$}
\end{minipage}
\end{center}
Let us note that $\mathfrak h_{p,N}(x)$ is 1--periodic; furthermore it is symmetric with respect to
$x=1/2$ and $\mathfrak h_{p,N}(1):= 0$. Therefore, it is enough to study this sum only on the interval $[1/2,1)$.
However, equivalently
   \begin{equation} \label{sum}
      \mathfrak h_{p,N}(x) =  \sin^p(\pi x)\sum_{k=N+1}^\infty\Bigg(\frac{1}{(k-x)^p}+\frac{1}{(k+x-1)^p} \Bigg)\, .
   \end{equation}
One introduces the {\em incomplete Lambda function} $\lambda(s;a)$ defined by the series
   \[ \lambda(s;a)= \sum_{n=1}^\infty \frac1{(2(n+a)-1)^s} \qquad \big( s>1,\ a \ge 0\big)\, .\]
\begin{theorem}\label{the2} For each $N$ there exists $p_*$ such that $\big(1/2, \mathfrak h_{p,N}(1/2)\big)$
is a local minimum of $\mathfrak h_{p,N}(x)$ if $p>p_*$, and it is a local maximum for $p<p_*.$
The value $p_*$ can be found as a solution of the equation
   \begin{equation} \label{dzeta}
      4(p_*+1)\lambda(p_*+2;N)-\pi^2\lambda(p_*;N)=0\,.
   \end{equation}
The lower and upper bounds for $p_*$ are:
   \begin{align} \label{p+}
      p_* &\ge \frac18 \Big(\pi^2(2N+1)^2-12+\sqrt{\left(\pi^2(2N+1)^2-12\right)^2-128}\,\,\Big)\,\\ \label{p-}
      p_* &\le \frac14 \sqrt{72\pi^2(N+1)^2(2\pi^2(N+1)^2-2-\pi^2) + 4 + 36\pi^2+9\pi^4}\nonumber \\
          & \qquad + 3\pi^2(N+1)^2 -\frac34\big(2+\pi^2\big)\,.
   \end{align}
\end{theorem}
\begin{proof} Let us consider
   \begin{align} \label{der}
      \mathfrak h_{p,N}'(x) &= p\sin^{p-1}(\pi x)(2x-1) \Bigg[\frac{\pi\cos(\pi x)}{2x-1}
                              \sum_{k=N+1}^\infty \Bigg( \frac1{(k-x)^p}+\frac1{(k+x-1)^p}\Bigg)\nonumber \\
                            & \qquad + \frac{\sin(\pi x)}{2x-1}
                              \sum_{k=N+1}^\infty \Bigg(\frac1{(k-x)^{p+1}}-\frac1{(k+x-1)^{p+1}}\Bigg)\Bigg]\, .
   \end{align}
Here, the termwise differentiation of series \eqref{sum} is legitimate because the series of derivatives
$\mathfrak h_{p,N}'(x)$ given by \eqref{der} converges uniformly.

Now, we will study \eqref{der} when $x\downarrow 1/2$. For all $x\in [1/2,1]$ we have that $p \sin^{p-1}(\pi x) (2x-1)\ge 0$.
Letting $x\downarrow 1/2$ the following limits appear in \eqref{der}:
   \begin{align*}
      \frac{\pi\cos(\pi x)}{2x-1} \sum_{k=N+1}^\infty \left(\frac{1}{(k-x)^p}+\frac{1}{(k+x-1)^p}\right)\,\,
                                    &\to \,\,-\sum_{k=N+1}^\infty \frac{\pi^2}{(k-1/2)^p}\,;\\
      \frac{\sin(\pi x)}{2x-1} \sum_{k=N+1}^\infty \left(\frac{1}{(k-x)^{p+1}}-\frac{1}{(k+x-1)^{p+1}}\right)\,\,
                                     &\to \,\, (p+1)\sum_{k=N+1}^\infty \frac{1}{(k-1/2)^{p+2}}\, .
   \end{align*}
Obviously, $\mathfrak h_{p,N}'(1/2)=0$. We are interested in the nature of the extremum at $x=1/2$.
Hence, from \eqref{der}, by L'H\^ospital rule one concludes
   \begin{equation} \label{sigh}
      \mathfrak h_{p,N}''(1/2) = 2\lim_{x\downarrow 1/2}\frac{\mathfrak h_{p,N}'(x)}{2x-1} =
              \sum_{k=N+1}^\infty \frac{2p(p+1)}{(k-1/2)^{p+2}}-\sum_{k=N+1}^\infty \frac{2p\pi^2}{(k-1/2)^p}\, .
   \end{equation}
To investigate \eqref{sigh} we take the following auxiliary result. Namely, let $\kappa \in \mathbb R_+\setminus \{1\}$ and
let $a:\mathbb{R}_+\to \mathbb{R}_+$ be a positive monotone increasing function. Then
   \begin{equation} \label{MathieuX}
      \int_{a(1)}^\infty\frac{[a^{-1}(x)]}{x^\kappa}\,{\rm d}x = \frac1{\kappa-1}\sum_{n=1}^\infty\frac1{a^{\kappa-1}(n)},
   \end{equation}
where $[\alpha]$ denotes the integer part of $\alpha$ and $a^{-1}(x)$ is the inverse function to $a(x)$, \cite[\S 8]{Pog4}. Now,
let us rewrite \eqref{sigh} in the form
   \[\mathfrak h_{p,N}''(1/2) = 2p(p+1)\sum_{n=1}^\infty\frac{1}{(n+N-1/2)^{p+2}}
                                         - 2p\pi^2\sum_{n=1}^\infty\frac{1}{(n+N-1/2)^{p}}\,.\]
By specifying $a(x)\equiv x+N-1/2\,,a^{-1}(x)=x-N+1/2$ we find from \eqref{MathieuX} that
   \begin{equation} \label{R13}
      \mathfrak h_{p,N}''(1/2) = 2p^2 \int_{N+1/2}^{\infty}\frac{[x-N+1/2]}{x^{p+1}}
                                 \Bigg(\frac{(p+2)(1+p^{-1})}{x^2}-\pi^2\Bigg){\rm d}x\,.
   \end{equation}
Hence, if
   \begin{equation} \label{N}
      \frac{(p+2)(1+p^{-1})}{(N+1/2)^2}-\pi^2< 0,
   \end{equation}
the inequality $\mathfrak h_{p,N}''(1/2)<0$ shows that the value $\mathfrak h_{p,N}(1/2)$ is a local maximum
of $\mathfrak h_{p,N}(x)$. The inequality \eqref{N} holds for $p=1$, the expression $(p+2)(1+p^{-1})$
decreases if $p\in [1,\sqrt{2})$ and increases if $p\in (\sqrt{2},\infty)$. Solving \eqref{N} with respect to $p$ we get
   \[p< \frac18 \left(\pi^2(2N+1)^2-12+\sqrt{\left(\pi^2(2N+1)^2-12\right)^2-128}\right)\,.\]
Thus, if $p_*$ exists, the lower bound on $p_*$ follows from the last inequality.

Now, we are looking for values of $p, N$ which ensure a local minimum at $x=1/2$. Considering \eqref{R13} once more,
it is sufficient to find when
   \begin{equation} \label{+}
      \int_{N+1/2}^{\infty}\frac{[x-N+1/2]}{x^{3}}\Bigg(\frac{(p+2)(1+p^{-1})}{x^2}-\pi^2\Bigg)
            \frac{{\rm d}x}{x^{p-2}}>0\,.
   \end{equation}
Fixing $N\in \mathbb N$ we remark that for some $p_0>\sqrt{2}$ there holds
   \begin{equation} \label{p*}
      (p+2)(1+p^{-1})\int_{N+1/2}^{\infty}\frac{[x-N+1/2]}{x^{5}}\,{\rm d}x>
      \pi^2\int_{N+1/2}^{\infty}\frac{[x-N+1/2]}{x^{3}}\,{\rm d}x \, ,
\end{equation}
when $p>p_0$, since both integrals are positive and finite, and $(p+2)(1+p^{-1})$ increases in $p$ when $p>\sqrt{2}$. Hence
   \begin{equation} \label{min}
      \int_{N+1/2}^{\infty}\frac{[x-N+1/2]}{x^{3}}\left(\frac{(p+2)(1+p^{-1})}{x^2}-\pi^2\right){\rm d}x>0 \qquad
      \big(p > p_0\big).
   \end{equation}
Indeed, the integrand in \eqref{min} changes the sign from positive to negative only once on $\mathbb{R}_+$.
Therefore, multiplying the integrand by the decreasing function $x^{2-p}$ does not change the sign of the
integral \eqref{min} for $p>\max(p_0,2)$. Thus \eqref{+} definitely holds for all $p\ge 2$ which satisfy \eqref{p*}.

Now, we have
   \begin{equation} \label{AB}
      \frac{ \int_{N+1/2}^{\infty} \frac{[x-N+1/2]}{x^{3}}\,{\rm d}x}{ \int_{N+1/2}^{\infty}\frac{[x-N+1/2]}{x^{5}}\,{\rm d}x} <
      \frac{ \int_{N+1/2}^{\infty}\frac{x-N+1/2}{x^{3}}\,{\rm d}x}{ \int_{N+1/2}^{\infty}\frac{x-N-1/2}{x^{5}}\,{\rm d}x}
          = \frac{3}{2}(4(N+1)^2-1) =: A(N)\,
   \end{equation}
and it follows easily that \eqref{+} holds for any $p$ satisfying the inequality
   \begin{equation} \label{pN}
      B(p):= (p+2)(1+p^{-1})>\frac{3\pi^2(4(N+1)^2-1)}{2}\, .
   \end{equation}
Indeed, bearing in mind the equivalent form of \eqref{AB}, that is
   \[ \int_{N+1/2}^{\infty}\frac{[x-N+1/2]}{x^{5}}\, {\rm d}x > \frac1{A(N)}\int_{N+1/2}^{\infty}\frac{[x-N+1/2]}{x^{3}}\, {\rm d}x\, ,\]
and re--writing \eqref{min} in the form
   \begin{align*}
      B(p)\,\int_{N+1/2}^{\infty}&\frac{[x-N+1/2]}{x^{5}}\, {\rm d}x - \pi^2\, \int_{N+1/2}^{\infty}\frac{[x-N+1/2]}{x^{3}}\,{\rm d}x \\
            & \qquad  > \Bigg( \frac{B(p)}{A(N)}- \pi^2\Bigg) \int_{N+1/2}^{\infty}\frac{[x-N+1/2]}{x^{3}}\, {\rm d}x >0\,,
   \end{align*}
we deduce the condition \eqref{pN}. (Moreover, $p>2$ follows from this inequality). Thus, solving the inequality
\eqref{pN} with respect to $p$ we conclude that, since $\mathfrak h_{p,N}''(1/2)>0$ for
   \begin{align*}
      p &\ge \frac14 \,\sqrt{72\pi^2(N+1)^2(2\pi^2(N+1)^2-2-\pi^2)+4+36\pi^2+9\pi^4} \\
        &\qquad \qquad + 3\pi^2\,\Big((N+1)^2 - \frac14\Big) - \frac32 \, ,
   \end{align*}
the function $\mathfrak h_{p,N}(x)$ possesses a local minimum at $x=1/2.$ So, if $p_*$ exists, the upper
bound on $p_*$ clearly follows from the last inequality.

Now, we prove the existence of $p_*.$  From the previous considerations it follows that for some values of
$p$ there is a local maximum in $x=1/2,$ for others there is a local minimum. Therefore, due
to the continuity at $p$ of the left part of (\ref{+}), we find that there exists at least one $p_*$ for which
   \[ \int_{N+1/2}^{\infty}\frac{[x-N+1/2]}{x^{p_*+1}}
           \Bigg(\frac{(p_*+2)(1+p_*^{-1})}{x^2}-\pi^2\Bigg){\rm d}x = 0. \]
By \eqref{N} we have that $p_*>\sqrt 2$. On the other hand as $(p+2)(1+p^{-1})$ increases, for $p>p_*$ we have
   \[ \int_{N+1/2}^{\infty}\frac{[x-N+1/2]}{x^{p_*+1}}\Bigg(\frac{(p+2)(1+p^{-1})}{x^2}-\pi^2\Bigg){\rm d}x\ge 0\,.\]
As the integrand changes the sign from positive to negative only once on $\mathbb{R}_+$ multiplying
the integrand by the decreasing function $x^{p_*-p}$ we get
   \[ \int_{N+1/2}^{\infty}\frac{[x-N+1/2]}{x^{p+1}}\Bigg(\frac{(p+2)(1+p^{-1})}{x^2}-\pi^2\Bigg){\rm d}x> 0\,.\]
Therefore, if $x=1/2$ is the abscissa of a local minimum for some $p_1$ then for all $p>p_1$ there is a
local minimum at $x=1/2$ as well. This shows the uniqueness of $p_*$ which is a root of \eqref{sigh}.
The final step in the proof is to notice that the
right--hand expression in \eqref{sigh} can be represented as
   \[ \sum_{k=N+1}^\infty \frac{2p(p+1)}{(k-1/2)^{p+2}}-\sum_{k=N+1}^\infty \frac{2p\pi^2}{(k-1/2)^p} =
        2^{p+1}p\big( 4(p+1) \lambda(p+2;N)-\pi^2 \lambda(p;N)\big)\, .\]
This finishes the proof of Theorem 3.
\end{proof}

\begin{corollary}\label{3.1}
   \begin{itemize}
      \item[a)] $\big(1/2, \mathfrak h_{p,N}(1/2)\big)$ is a local maximum by {\rm (\ref{N})} if
         \[ N \ge \frac1\pi \sqrt{(p+2)(1+p^{-1})}-\frac{1}{2}\,.\]
      \item[b)] By the previous estimate $\big(1/2, \mathfrak h_{p,N}(1/2)\big)$ is a local maximum
                for arbitrary $N\in\mathbb N$ if $ \sqrt{(p+2)(1+p^{-1})}\le 3\pi/2,$ i.e.
         \[p\le \frac18 \left(9 \pi^2+\sqrt{16-216\pi^2+81\pi^4}-12\right)\approx 19.1019\, .\]
   \end{itemize}
\end{corollary}
\begin{remark} \textit{The case {\rm b)} in {\rm Corollary \ref{3.1}} covers the particular case $p=2$ discussed in \cite{OP1}.}
\end{remark}
We need some bounds for $p_*$ in solving numerically the
transcendental equation (\ref{dzeta}).
\begin{remark} \label{rem1} \textit{For certain $N$ the corresponding values of $p_*$ and their bounds
given in Theorem~{\rm\ref{the2}}, are shown in the table:}
\begin{center}\begin{tabular}{|c|c|c|c|c|}
  \hline
  % after \\: \hline or \cline{col1-col2} \cline{col3-col4} ...
  \rule[-2mm]{0mm}{0.8cm}$N$ & 1 & 2 & 3 & 4 \\
  \hline
  \rule{0mm}{0.6cm}$p_*$ & 21.2069 & 60.685 & 119.903 & 198.859\\
  lower bound \eqref{p+} & 19.1019 & 58.6509 & 117.8857 & 196.8493\\
  \rule[-3mm]{0mm}{0.5cm}upper bound \eqref{p-} & 219.057 & 515.1504 & 929.6755 & 1462.6349\\
  \hline
\end{tabular}\end{center}
\textit{One can see that the lower bounds correspond to exact values of $p_*\,.$ However, the upper bounds are rough.
But they can be sharpened by, for example, following the previous method of calculation. One considers
   \[(p+2)(1+p^{-1})\int_{N+1/2}^{\infty}\frac{[x-N+1/2]}{x^{4+\varepsilon}}\, {\rm d}x>
                   \pi^2\int_{N+1/2}^{\infty}\frac{[x-N+1/2]}{x^{2+\varepsilon}}\, {\rm d}x\]
instead of {\rm \eqref{p*}} where $p-1>\varepsilon>0$, then minimizes the value of the upper bound for $p_*$
with respect to some admissible $\varepsilon.$}
\end{remark}
\begin{remark}\textit{Due to numerical results in {\rm Remark~{\rm\ref{rem1}}} the estimate in {\rm Corollary \ref{3.1} b)} can be
improved to $p\le 21.2069\,.$}
\end{remark}
Now, we are interested in properties of \eqref{sum} at an arbitrary point $x$. Let us denote the $k$--th term of
(\ref{sum}) by $\psi_k(x), k\ge 2$, i.e.
   \[ \mathfrak h_{p,N}(x)= \sin^p(\pi x)\sum_{k=N+1}^\infty \Bigg(\frac1{(k-x)^p}+\frac{1}{(k+x-1)^p}\Bigg)
                            =: \sum_{k=N+1}^\infty \psi_k(x)\,.\]
\begin{lemma}\label{th3} For $x\in(1/2,1)$, from $\psi_k'(x)<0$ it follows that $\psi_{k+1}'(x)<0.$
\end{lemma}
\begin{proof}
Using the well-known identity
   \[ \frac{\Gamma(s)}{\alpha^s}=\int_0^\infty e^{-\alpha t}t^{s-1}\, {\rm d}t\, ,\]
specifying $\alpha=k-x$, then $k+x-1$ we get
   \begin{align} \label{psi'}
      \psi_k'(x) &= p\sin^{p-1}(\pi x)\Bigg[\pi\cos(\pi x)\Bigg(\frac1{(k-x)^p}+\frac1{(k+x-1)^p}\Bigg) \nonumber \\
                 &  \qquad + \sin(\pi x) \Bigg(\frac1{(k-x)^{p+1}}-\frac1{(k+x-1)^{p+1}}\Bigg)\Bigg] \\
                 &= \frac{p\sin^{p-1}(\pi x)}{\Gamma(p)}\int_0^\infty e^{-kt}t^{p-1}\Big(\pi\cos(\pi x)
                    \big(e^{xt}+e^{(1-x)t}\big)\nonumber \\
                 &  \qquad + \frac{t\sin(\pi x)}{p}\left(e^{xt}-e^{(1-x)t}\right)\Big)dt
                    =:\int_0^\infty g_k(t)\, {\rm d}t\, .\nonumber
   \end{align}
Therefore
   \[\psi_{k+1}'(x)=\int_0^\infty e^{-t}g_k(t)\, {\rm d}t.\]
Denote
   \[ \widetilde{g}_k(t) := \pi\cos(\pi x) \big(e^{xt}+e^{(1-x)t}\big)
                           + \frac{t\sin(\pi x)}{p}\left(e^{xt}-e^{(1-x)t}\right).\]
Let us show that there exists $t_0>0$ such that $\widetilde{g}_k(t)\le 0,\ 0\le t\le t_0,$ and
$\widetilde{g}_k(t)\ge 0,\ t\ge t_0.$ To do this, let us note that $\widetilde{g}_k(0) = 2\pi\cos(\pi x)< 0,\
x\in(1/2,1).$ Then rewrite $\widetilde{g}_k(t)$ in the form
   \[ \widetilde{g}_k(t)= p^{-1}t\sin(\pi x)\left(e^{xt}+e^{(1-x)t}\right)\left(\pi p\,t^{-1}\cot(\pi x)
                          +\frac{2}{1+e^{(1-2x)t}}-1\right).\]
The term $p^{-1}t\sin(\pi x)\left(e^{xt}+e^{(1-x)t}\right)$ is positive for all $x\in(1/2,1),\ t> 0.$
For fixed $x\in(1/2,1)$ the function $\pi p\,t^{-1}\cot(\pi x)$ increases from $-\infty$ to
$0$ and the function $\frac{2}{1+e^{(1-2x)t}}-1$ increases from $0$ to $1$ when $t$ runs over the positive half--axis.
Hence, there exists a unique $t_0$ such that $\widetilde{g}_k(t_0)=0$.

By the relation
   \[ g_k(t) = \frac{p\sin^{p-1}(\pi x)}{\Gamma(p)}e^{-kt}t^{p-1}\widetilde{g}_k(t)\]
we deduce that $g_k(t)\le 0,\ 0\le t\le t_0,$ and $g_k(t)\ge 0,\ t\ge t_0$.
Therefore, since $\psi_k'(x)=\int_0^\infty g_k(t)\, {\rm d}t<0$ we conclude that $\psi_{k+1}'(x)=\int_0^\infty e^{-t}g_k(t)\, {\rm d}t<0$
since $e^{-t}$ is a decreasing function.
\end{proof}
For further investigations we need the next two lemmata.
\begin{lemma}\label{l1} For each $z\in[0,1/2]$ we have
   \begin{equation} \label{sine1}
      \frac{\sin(2\pi z)}{2\pi z}\le \frac{1-z^2}{1+z^2}\,.
   \end{equation}
\end{lemma}
\begin{proof}
From the series expansion for the sine function we have
   \[ \frac{\sin(2\pi z)}{2\pi z}\le 1-\frac{(2\pi z)^2}{3!}+\frac{(2\pi z)^4}{5!}\,,\]
as a consequence of $\frac{(2\pi z)^{2k}}{(2k+1)!}\ge\frac{(2\pi z)^{2k+2}}{(2k+3)!},\,k\in \mathbb N$,
which holds for all $z\in[0,1/2]$. Thus, to prove \eqref{sine1}, it is enough to show that
   \begin{equation}\label{tri1}
      1-\frac{(2\pi z)^2}{3!}+\frac{(2\pi z)^4}{5!}\le\frac{1-z^2}{1+z^2}\,.
   \end{equation}
On putting $z^2=t$, \eqref{tri1} can be rewritten as
   \[ \frac{t}{1+t} \le \frac{\pi^2t}{3}-\frac{\pi^4t^2}{15}\,,\]
or
   \[ 5\pi^2-15+(5\pi^2-\pi^4)t-\pi^4t^2\ge 0. \]
For $t\in[0,1/4]$ the last inequality holds true, because the polynomial $5\pi^2-15+(5\pi^2-\pi^4)t-\pi^4t^2$ has
two real roots at $t_1\approx -0.9,$ and $t_2\approx 0.39\,.$
\end{proof}
\begin{lemma}\label{lem2} The absciss{\ae} of extrema of $\psi_k(x)$ are $1/2$ and $1$ for all admissible $p$. Moreover,
if $p>\pi^2(k-1/2)^2-1$, $\psi_k(x)$ attains an extremum at some $x_k^*\in(1/2,A_k),\, A_k\in(1/2,1)$ being the
unique solution of the equation
   \[\cot(\pi x)=-\frac{1}{\pi(k-x)}.\]
Each $x_k^*\in(1/2,A_k)$ is the abscissa of an extremum corresponding to a unique $\widetilde{p}$ which satisfies
   \[\widetilde{p} = \frac{\ln \displaystyle \frac{ 1-\pi\cot(\pi x_k^*)(k+x_k^*-1)}
                          {\displaystyle \pi\cot(\pi x_k^*)(k+x_k^*-1)-1+\frac{2k-1}{k-x_k^*}}}
                          {\displaystyle \ln(k+x_k^*-1) -\ln(k-x_k^*)}>1\,.\]
\end{lemma}
\begin{proof}
We immediately obtain by \eqref{psi'} that $\psi_k'(x)$ vanishes at $1/2$ and $1$. Another possible extremum
could arise from solutions of the equation
   \begin{equation} \label{pp}
      \pi\cot(\pi x)\left(\frac{1}{(k-x)^p}+\frac{1}{(k+x-1)^p}\right)=
                             \frac{1}{(k+x-1)^{p+1}}-\frac{1}{(k-x)^{p+1}}
   \end{equation}
where $x\in(1/2,1)$. Let us show that \eqref{pp} has one only solution. Simple transformations of (\ref{pp}) give
   \begin{equation} \label{eq}
      \frac{(k+x-1)^{p+1}-(k-x)^{p+1}}{((k+x-1)^{p}+(k-x)^{p})(x-1/2)} =
                              \frac{-\pi\cot(\pi x)(k+x-1)(k-x)}{(x-1/2)}\,.
   \end{equation}
On putting $z=x-1/2\in (0,1/2),$ and $a=k-1/2>0$, (\ref{eq}) becomes
   \begin{equation} \label{eq1}
      L(z) := \frac{(a+z)^{p+1}-(a-z)^{p+1}}{z((a+z)^{p}+(a-z)^{p})}=\frac{\pi\tan(z\pi)(a^2-z^2)}{z} =: R(z)\,.
   \end{equation}
We now differentiate the left-hand member of \eqref{eq1} to obtain
   \begin{equation} \label{deri}
      L'(z) = -\,\frac{\displaystyle \left(\frac{a+z}{a-z}\right)^{2p}-1
             - p\big(\left(\frac{a+z}{a-z}\right)^2-1\big)\left(\frac{a+z}{a-z}\right)^{p-1}}
               {\displaystyle z^2(\left(\frac{a+z}{a-z}\right)^{p}+1)^2}\,.
   \end{equation}

Now, it is not hard to show that for each $t\ge 1,\ p\ge 1$ we have
   \begin{equation} \label{pow}
      t^{2p}-1\ge p(t^2-1)t^{p-1}.
   \end{equation}
On taking $t=\frac{a+z}{a-z}$ in \eqref{deri}, inequality \ref{pow} shows that $L(z)$ decreases in $z.$

We now differentiate the right-hand member of \eqref{eq1} to obtain
   \begin{align} \label{der2}
      R'(z) &= \frac{\pi(a^2+z^2)}{z\cos^2(\pi z)} \left(1-\frac{2z^2}{a^2+z^2}-\frac{\sin(2\pi z)}{2\pi z}\right)
               \nonumber \\
            &\ge \frac{\pi(a^2+z^2)}{z\cos^2(\pi z)}\left(1-\frac{2z^2}{1+z^2}-\frac{\sin(2\pi z)}{2\pi z}\right)\,.
   \end{align}
It follows from (\ref{der2}) and Lemma~\ref{l1} that $R(z)$ increases in $z$.

Letting $z\to 1/2$ we get
   \[ \lim_{z \to 1/2} L(z) = 2\frac{(a+1/2)^{p+1}-(a-1/2)^{p+1}}{(a+1/2)^{p}+(a-1/2)^{p}}, \qquad
      \lim_{z \to 1/2} R(z) = \infty\, ;\]
and similarly
   \[ \lim_{z \to 0} L(z) = p+1, \qquad \lim_{z \to 0} R(z) = (a\pi)^2\, .\]
By the aforementioned properties of $L(z)$ and $R(z)$ we find that there are only two possibilities: {\bf (i)}
\eqref{eq1} has no solutions, or {\bf (ii)} it has a unique solution on the interval $(0,1/2)$. Moreover,
\eqref{eq1} has a solution on the interval $z\in (0,1/2)$ if and only if $p+1 >\pi^2a^2=\pi^2(k-1/2)^2.$

Further simple transformations of (\ref{pp}) give
   \begin{equation} \label{p}
      \left(1+\frac{2x-1}{k-x}\right)^p=\frac{1-\pi\cot(\pi x)(k+x-1)}{\pi\cot(\pi x)(k+x-1)-1+\frac{2k-1}{k-x}}.
   \end{equation}
The expression on the left side of (\ref{p}) increases unboundedly with growing $p$. Therefore, $x_k^*\in(1/2,1)$
might be an extremal point for one $p$ only. From \eqref{p} we easily calculate this value
   \begin{equation} \label{tilda}
      \widetilde{p} = \frac{\ln \frac{ \displaystyle 1-\pi\cot(\pi x_k^*)(k+x_k^*-1)}
                           { \displaystyle \pi\cot(\pi x_k^*)(k+x_k^*-1)-1
                      + \frac{2k-1}{k-x_k^*}}}{\ln\left(1+ \displaystyle \frac{2x_k^*-1}{k-x_k^*}\right)}\,.
   \end{equation}
Let us show that $x_k^*$ cannot be greater than $A_k.$ The denominator in (\ref{tilda}) is defined on the whole of
$(1/2,1)$. To investigate the numerator we study the equation
   \[ \pi \cot(\pi x)(k+x-1)-1+\frac{2k-1}{k-x} = 0,\]
that is
   \begin{equation} \label{A}
      \cot(\pi x)=-\frac{1}{\pi (k-x)}.
   \end{equation}
Because of
   \[0=\cot(\pi/2)>-\frac{1}{\pi (k-1/2)},\quad -\infty=\lim_{x\to 1}\cot(\pi x)<-\frac{1}{\pi (k-1)}\, ,\]
and
   \[\cot'(\pi x)< \left(-\frac{1}{\pi (k-x)}\right)'\qquad \big( x\in(1/2,1)\big)\]
it is obvious that the equation \eqref{A} has an unique solution in $(1/2,1)$ which we denote by $A_k$. Of
course, $x_k^*\in (1/2,A_k)$, because $\pi\cot(\pi x )(k+x -1)-1+\frac{2k-1}{k-x }<0$ for $x>A_k$.

Let us show that for arbitrary $x\in(1/2,A_k)$ we have
   \begin{equation} \label{A'}
      \ln\left(-1+\frac{\frac{2k-1}{k-x}}{\pi\cot(\pi x)(k+x-1)-1+\frac{2k-1}{k-x}}\right)>0
   \end{equation}
and the value $\widetilde{p}$, given above by \eqref{tilda}, is greater than 1.

Because of \eqref{tilda}, \eqref{A'} becomes equivalent to
   \[ -1 + \frac{\frac{2k-1}{k-x }}{\pi\cot(\pi x )(k+x -1)-1+\frac{2k-1}{k-x }}>1+\frac{2x -1}{k-x },\]
that is, to
   \begin{equation} \label{poz}
      2-\frac{2k-1}{k-x }> \pi\cot(\pi x )(k+x -1)\, .
   \end{equation}
For $x=1/2$, \eqref{poz} becomes an identity. Note that for $x\in(1/2,A_k)$ the derivatives of both sides of
\eqref{poz} satisfy the inequality
   \[\left(2-\frac{2k-1}{k-x}\right)'=-\frac{2k-1}{(k-x)^2}>\pi\left(\cot(\pi x)-\frac{\pi(k+x-1)}
        {\sin^2(\pi x)}\right)=\pi \left(\frac{k+x-1}{\tan(\pi x )}\right)' \]
which is satisfied since
   \[\sin^2(\pi x)<\frac{\pi^2}{2}\cdot \frac{k+x-1}{k-1/2}\cdot(k-x)^2 \qquad \big(x\in(1/2,1),\ k\ge 2\big)\, . \]
Therefore, $\widetilde{p}$ is correctly defined.
\end{proof}
\begin{corollary}\label{th5} For $p \le \pi^2(k-1/2)^2-1$ the function $\psi_k(x)$ decreases on $[1/2,1]$.
\end{corollary}
\begin{proof} The statement is a consequence of Lemma~\ref{lem2} and the fact $\psi_k(1/2)>\psi_k(0)=0$.
\end{proof}
\begin{theorem}\label{th2} For $p\le \pi^2(N+1/2)^2-1$ we have
   \begin{equation} \label{XX}
      \sum_{n \in \mathbb Z\setminus \mathfrak J_x} |{\rm sinc}(x-n)|^p \le
                             2\left(\frac{2}{\pi}\right)^p\lambda(p;N)\,.
   \end{equation}
The inequality is sharp and becomes an equality at $x=1/2$.
\end{theorem}
\begin{proof} The assertion is an easy consequence of Lemma~\ref{th3}, Corollary~\ref{th5} and the identity
   \[ \sum_{n \in \mathbb Z\setminus \mathfrak J_{1/2}} |{\rm sinc}(1/2-n)|^p =
            2\left(\frac{2}{\pi}\right)^p\sum_{k=N+1}^\infty\frac{1}{(2k-1)^p}=
            2\left(\frac{2}{\pi}\right)^p\lambda(p;N)\,.\]
\end{proof}
\begin{remark} \textit{The result of {\rm Theorem}~{\rm\ref{th2}} improves the lower bound {\rm \eqref{p+}} in
{\rm Theorem}~{\rm\ref{the2}}. For some positive integer $N$ the numerical values of
$\pi^2(N+1/2)^2-1$ are shown in the table (compare with the table
in {\rm Remark}~{\rm\ref{rem1}}):}
\begin{center}\begin{tabular}{|c|c|c|c|c|}
  \hline
  % after \\: \hline or \cline{col1-col2} \cline{col3-col4} ...
  \rule[-2mm]{0mm}{0.8cm}$N$ & 1 & 2 & 3 & 4 \\
  \hline
  \rule[-2mm]{0mm}{0.8cm}$\pi^2(N+1/2)^2-1$ & 21.2066 & 60.6849 & 119.902 & 198.859\\
  \hline
\end{tabular}\end{center}\end{remark}
\begin{remark} \textit{The expression $\pi^2(N+1/2)^2-1$ grows quickly in $N$ and at the same time covers a
wide range of small $p$ values. Therefore, the estimate \eqref{XX} is important for applications.}
\end{remark}

\section{\bf Upper bounds without known signal decay rate}
Here, and in what follows let $(p,q)$ be a positive H\"older pair, i.e.
$p^{-1}+q^{-1} = 1,$ $p >1.$

The most frequently appearing  rearrangement of \eqref{3} in the literature is of the form
   \begin{equation} \label{4}
      \| T_{\mathfrak J}(f;\mathbf{x})\|_\infty \le
         \Big( \sum_{\mathbf n \in \mathbb Z^d\setminus \mathfrak J}|S(\mathbf x,t_{\mathbf n})|^p\Big)^{1/p}
         \Big( \sum_{\mathbf n \in \mathbb Z^d\setminus \mathfrak J}|f(t_{\mathbf n})|^q\Big)^{1/q}
         =: A_p\,B_q.
   \end{equation}
To make the approximant $Y_{\mathfrak J}(f;\mathbf{x})$ more precise it is of interest to assume
$\mathfrak J = \mathfrak J_\mathbf x$, i.e. that the index set of sampling restoration $\mathfrak
J_{\mathbf x}$ depends on the location of $\mathbf x$ with respect to the behaviour of $f$ when estimating $B_q$.
That means $T_{\mathfrak J_\mathbf x}(f;\mathbf{x})$ is {\em a fortiori} time shifted and
possesses time adapted sampling size. Thus our approach is closer to an interpolation rather
than to extrapolation procedure. The earlier works \cite{FLS, OP1, Yen} did not mention these facts.

Using the classical approach one operates with the straightforward estimate $B_q \le C_{f,\mathfrak J}\,\|f\|_q$
where $C_{f, \mathfrak J}$ is a suitable absolute constant. Therefore \eqref{4} becomes
   \[ \| T_{\mathfrak J_\mathbf x}(f;\mathbf{x})\|_\infty \le A_p\,C_{f,\mathfrak J_\mathbf x}\, \|f\|_q\,,\]
where $C_{f,\mathfrak J_\mathbf x}$ depends on the properties of the signal function and
vanishes when $|\mathfrak J_{\mathbf x}| \to \infty$.

To obtain a class of truncation error upper bounds when the decay rate of the initial signal function
is not known we are interested in estimates for $A_p$ which vanish with $|\mathfrak J_{\mathbf x}| \to \infty,$
and estimates like $C_{f,\mathfrak J_\mathbf x}\le C_{f}$ which are valid for all $\mathbf x.$

We will consider the case $\mathfrak T=\mathbb Z^d.$ Let us specify the sampling function
   \[ S(\mathbf x, t_{\mathbf n}) = \prod_{j=1}^d {\rm sinc}(x_j-n_j)\,,\]
where ${\mathbf n}=(n_1,...,n_d)\in \mathbb{Z}^d\,.$
Then (\ref{1}) becomes the equally sampled (regular) uniformly
convergent WKS sampling formula
   \begin{equation} \label{9}
      f(\mathbf{x}) = \sum_{\mathbf{n}\in \mathbb{Z}^d}f(\mathbf{n}) \prod_{j=1}^d {\rm sinc}(x_j-n_j)
                      \qquad \big(\mathbf x \in \mathbb R^d\big)\,.
   \end{equation}
Let $\mathbf N:= (N_1, \cdots, N_d)\in \mathbb N^d,\,\mathfrak J_{\mathbf x} := \{ \mathbf n:\,
\bigwedge_{j = 1}^d(|x_j - n_j|\le N_j)\}$. Then the WKS sum truncated to $\mathfrak J_{\mathbf x}$ reads
   \[ Y_{\mathfrak J_{\mathbf x}}(f;\mathbf{x}) = \sum_{\mathbf n \in \mathfrak J_{\mathbf x}}
                     f(\mathbf{n})\prod_{j=1}^d {\rm sinc}(x_j-n_j)\]
and one introduces the consequent truncation error by
   \[ \|T_{\mathbf{N},d}(f;\mathbf x)\|_\infty :=\|f(\mathbf x)-Y_{\mathfrak J_{\mathbf x}}(f;\mathbf{x})\|_\infty\,.\]
In the sequel let us define $\widetilde{N}:=\min_{j=1,\cdots,d}N_j$ and
$\boldsymbol \pi=(\pi,\ldots,\pi)_{1\times d}.$
  \begin{theorem}\label{th4} Let $f\in B_{\boldsymbol \pi,d}^q,\ q\ge 1+\big(\pi^2(\widetilde{N}+1/2)^2-2\big)^{-1}.$
Then we have
   \begin{equation} \label{12}
      \|T_{\mathbf{N},d}(f;\mathbf x)\|_\infty
             \le C(\mathbf{N},d,q)
              \cdot \|f\|_q\, ,
   \end{equation}
where
   \[C(\mathbf{N},d,q):= \frac{2^{2-\frac{1}{q}}}{\pi}\,\mathfrak C_{\frac{q}{q-1}}^{(d-1)(1-\frac{1}{q})}\,
                         \mathfrak B_{d,q}^{1/q}\, \left(\sum_{j=1}^d
                         \lambda\left(\frac{q}{q-1};N_j\right)\right)^{1-\frac{1}{q}}\, ;\]
and the constants $\mathfrak C_\alpha, \mathfrak B_{d, \beta}$ are introduced in {\rm Theorem 2} and by
{\rm \eqref{PP-dconst}} respectively.
\end{theorem}
\begin{proof} Now $f \in B_{\boldsymbol \pi,d}^q\,,$ hence, by Lemma~\ref{nik_i},
$f(x_1,...,x_{k-1},\cdot,x_{k+1},...,x_d) \in B_{\pi,1}^q$ for all $k=1,\cdots,d$. So, the direct $d$--dimensional
generalization of the proof given for the two--dimensional case in \cite[\S 48]{PPII} shows that \eqref{9} is valid,
that is
   \[ \|T_{\mathbf{N},d}(f;\mathbf x)\|_\infty = \Bigg\|\sum_{\mathbf{n}\in \mathbb{Z}^d\setminus \mathfrak
                   J_{\mathbf x}}f(\mathbf{n}) \prod_{j=1}^d {\rm sinc}(x_j-n_j)\Bigg\|_\infty\,.\]
Hence, it is not hard to see that for a H\"older pair $(p,q)$ we have
   \begin{equation} \label{13}
      \| T_{\mathbf{N},d}(f;\mathbf x)\|_\infty
                      \le \Bigg(\sum_{\mathbf n \in \mathbb Z^d \setminus \mathfrak J_{\mathbf x} }
                          \prod_{j=1}^d \bigl| {\rm sinc}(x_j-n_j)\bigr|^p\Bigg)^{1/p}\,
                          \Bigg(\sum_{\mathbf n \in \mathbb Z^d \setminus \mathfrak J_{\mathbf x} }|f(\mathbf n)|^q
                          \Bigg)^{1/q}\,.
   \end{equation}
By Theorem~\ref{th1} the multidimensional Plancherel--P\'olya inequality \eqref{PP-d} holds for functions $f$
belonging to the Bernstein space $B_{\boldsymbol \pi,d}^q,\, q \ge 1$ and since in the
regular sampling case one has $\delta_\ell=1, \ell = 1,\cdots,d$, that is $\mathfrak T \equiv \mathbb Z^d$. It
follows that the second term in \eqref{13} is bounded by $\mathfrak B^{1/q}_{d,q} \| f\|_q$. Indeed, we easily deduce
   \[ \sum_{\mathbf n \in \mathbb Z^d \setminus \mathfrak J_{\mathbf x} }|f(\mathbf n)|^q \le
      \sum_{\mathbf n \in \mathbb Z^d}|f({\mathbf n})|^q \le \mathfrak B_{d,q} \, \| f\|_q^q\, .\]
Now, we use the estimate
   \begin{align} \label{pq}
      \sum_{\mathbf n\in\mathbb Z^d\setminus\mathfrak J_{\mathbf x}} &\prod_{j=1}^d \bigl| {\rm sinc}(x_j-n_j)\bigr|^p\le
      \sum_{j=1}^d\sum_{n_j \in \mathbb Z \setminus \mathfrak J_{x_j}}\bigl|{\rm sinc}(x_j-n_j)\bigr|^p
      \prod_{k\not=j}\sum_{n_k \in \mathbb Z} \bigl| {\rm sinc}(x_k-n_k)\bigr|^p\nonumber \\
 &\le \left( \sup_{x\in[1/2,1]}\ \sum_{n\in\mathbb Z} \bigl| {\rm sinc}(x-n)\bigr|^p\right)^{d-1}
      \sum_{j=1}^d\sum_{n_j \in \mathbb Z \setminus \mathfrak J_{ x_j} } \bigl| {\rm sinc}(x_j-n_j)\bigr|^p\, .
   \end{align}
The application of \eqref{XX} to the last term in \eqref{pq} requires $p\le \pi^2(N_j+1/2)^2-1$ for all
$j=1,\cdots,d$. Rewriting all these constraints
we obtain $q\ge 1+\left(\pi^2(\widetilde{N}+1/2)^2-2\right)^{-1}.$ So, as $\sigma_\ell = \pi, \delta_\ell = 1$
we conclude
   \[ \sum_{\mathbf n \in \mathbb Z^d \setminus \mathfrak J_{\mathbf x}}
            \prod_{j=1}^d \bigl| {\rm sinc}(x_j-n_j)\bigr|^p \le 2\left(\frac{2}{\pi}\right)^p
            \sum_{j=1}^d \lambda(p;N_j)\Bigg(\sup_{x\in[1/2,1]}\ \sum_{n\in\mathbb Z} \bigl|
            {\rm sinc}(x-n)\bigr|^p\Bigg)^{d-1}.\]
By applying Theorem~\ref{the0} to this estimate, then substituting the values of all constants involved and
noticing that $p^{-1}+q^{-1} = 1, (p,q)$ being a positive H\"older pair, we complete the proof of the Theorem.
\end{proof}

\begin{remark}\textit{For the particular case of equal truncation sizes $N_1=...=N_d=\widetilde{N}$ and
$q_*:=1+\frac1{\pi^2(\widetilde{N}+1/2)^2-2 \rule{0mm}{0.33cm}}$ the
corresponding numerical values of $C_{\mathbf N}:=C(\mathbf{N},d,q)$ are given in
the table:}
\begin{center}\begin{tabular}{|c|c|c|c|c|c|c|c|c|c|}
  \hline
  % after \\: \hline or \cline{col1-col2} \cline{col3-col4} ...
  $\widetilde{N}$ & \multicolumn{3}{c}{\rule[-2mm]{0mm}{0.7cm}1} &\multicolumn{3}{|c|}{2} & \multicolumn{3}{c|}{3}\\
  \hline
  \rule[-1.5mm]{0mm}{0.6cm}$q_*$ & \multicolumn{3}{|c|}{$1.0495$}  & \multicolumn{3}{|c|}{$1.0168$}  & \multicolumn{3}{|c|}{$1.0084$} \\
    \hline
      \rule[-1.5mm]{0mm}{0.6cm} $d$ & 1 & 2 & 3 & 1 & 2 & 3 & 1 & 2 & 3\\
    \hline
     \rule[-2mm]{0mm}{0.7cm}$C_{\mathbf{N}}$ & 0.6727 & 2.1328 & 6.6703 & 0.3968  & 1.2369  & 3.8368 & 0.2822 & 0.8756 & 2.7103  \\
    \hline
\end{tabular} \end{center}
\end{remark}

\begin{remark} {\it In the case $p = q = 2$ the $L^2(\mathbb R^d)$--theory shows that the universal
truncation error upper bound takes the form
   \begin{align*}
      \|T_{\mathbf{N},d}(f,\mathbf x)\|_\infty  &\le  \sqrt{1-\prod_{j=1}^d\left(1-\frac{8}{\pi^{2}}
                              \lambda(2;N_j)\right)}\,\cdot\,\|f\|_2 \\
                          &= \sqrt{1-\frac{8^d}{\pi^{2d}}\prod_{j=1}^d\sum_{n_j=1}^{N_j}
                             \frac{1}{(2n_j-1)^2}}\,\cdot\,\|f\|_2\,.
   \end{align*}
This result is valid for all $\mathbf{N}\in \mathbb N^d$ and it is mainly simpler than the related restriction of
{\rm \eqref{12}}. Moreover, this bound is sharp, see \cite[Theorem 1]{OP2}.}
\end{remark}
\begin{corollary} Let $f\in B_{\boldsymbol \pi,d}^q,\ q>1$. The convergence rate in the multidimensional WKS sampling
restoration procedure is of magnitude
   \begin{equation}\label{speed}
      \| T_{\mathbf{N},d}(f;\mathbf x)\|_\infty = \mathscr O\big(\widetilde{N}^{-1/q}\big)
                  \qquad \big(\widetilde{N}\to\infty\big)\,.
   \end{equation}
\end{corollary}
\begin{proof}  By the estimate
   \[ \lambda(p;N) \le \frac1{(2N+1)^p} + \int_{N+1}^\infty \frac{{\rm d}x}{(2x-1)^p} =
                     \frac1{(2N+1)^{p-1}}\left(\frac1{2N+1}+\frac{1}{2(p-1)}\right) \]
we obtain $\lambda(p;N) = \mathscr O(N^{-p+1})\,.$  Letting $\widetilde{N}\to \infty$ under the constraint $q \ge 1+
\big(\pi^2(\widetilde{N}+1/2)^2-2\big)^{-1}$ we arrive at \eqref{speed} with the aid of Theorem~\ref{th4}.
\end{proof}

\section{\bf Final remarks}
\hspace{-4mm}{\sf A.} In the majority of articles devoted to Whittaker--Kotel'nikov--Shannon sampling theorems, their
numerous generalizations, the related truncation
error analysis and the sampling approximation convergence questions in the various $L^p(\mathbb R^d)$--type spaces
(e.g. \cite{hig}, \cite{li}, \cite{LF}, and the references therein), the authors consider particular classes of
functions with prescribed decay rates. These assumption upon the initial signal give one an opportunity
to estimate the constant $B_q$ in \eqref{4} so, that the estimates vanish when the finite sampling restoration sum
size parameter $N$ runs to infinity. The constant $A_p$ is another feature of interest here. It too was introduced in
\eqref{4}, and its novelty lies in the fact that one usually estimates by an absolute constant which does
not depend on $N$.

\hspace{-4mm}{\sf B.} In this article we propose estimates for $A_p$
which depend on $N$ and tends to zero for $N$ growing. This approach
enables us to consider and to obtain approximation error estimates
for wide functional classes without strong assumptions upon the decay rate of
initial signals. It seems that proposed estimates
are the first attempts in sampling theory to analyze the properties
of the series
   \[\mathfrak h_{p,N}(x):= \sum_{n\not\in\mathfrak J_x} |{\rm sinc}(x-n)|^p\]
for $p\not=2$ as functions of $x$. The proposed procedure treats the
truncation error $T_{\mathfrak J_{\mathbf x}}(f;\mathbf x)$ in one
and higher dimensions and it results in sharper estimates in
comparison with known ones in the literature. Very interesting
relations have been found on the dependence between the extremums of
$\mathfrak h_{p,N}(x)$, its various properties and the joint
behaviour of $N$ and $p$. New related truncation error upper bounds are
derived in $\|\cdot\|_\infty$--norm for the sampling restoration
procedure. Several numerical examples illustrates the matter in the
multidimensional case too.

\hspace{-4mm}{\sf C.} Moreover, the results obtained and the numerical simulations suggest some new hypotheses
and open problems:
\begin{enumerate}
  \item To investigate the case of $p_*$ in Theorem~\ref{the2}.
  \item It seems that there are only two possible cases in the behaviour of $\mathfrak h_{p,N}(x),\,x\in [1/2,1]$.
    The first one concerns decreasing $\mathfrak h_{p,N}(x)$, its maximum is in $x=1/2$; the second one means
    increasing $\mathfrak h_{p,N}(x)$ to certain local maximum, and then it decreases,
    see the plots of $\mathfrak h_{2,2}(x), \mathfrak h_{27,2}(x)$ on Figures~1,\,2. Prove or disprove!
  \item To estimate the maximum value of $\mathfrak h_{p,N}(x)$ in the second case mentioned above.
  \item To obtain sharp estimates in Theorem~\ref{th4}. It has to be mentioned that for $p=2$ such sharp
    estimates were derived in \cite{OP2}.
  \item To apply the results obtained to the discontinuous signals, see \cite{But}.
  \item To apply the results obtained to the stochastic case, see \cite{OP}.
\end{enumerate}

\bigskip

\bigskip
\begin{center}
\textbf{Acknowledgements}
\end{center}
The authors thank Professor Emeritus J.R. Higgins for numerous valuable comments, discussion and suggestions
that have helped to improve our paper.

The recent investigation was supported in part by Research Project No.
112-2352818-2814 of Ministry of Sciences, Education and Sports of Croatia and in part by
La Trobe University Research Grant--501821 "Sampling, wavelets and optimal stochastic modelling".

\end{document}